\documentclass[a4paper,11pt]{article}

\usepackage{amsfonts}
\usepackage{amsmath}
\usepackage{amsthm}
\usepackage{array}
\usepackage[nobreak]{cite}
\usepackage[usenames,dvipsnames]{color}
\usepackage{complexity}
\usepackage{environ}
\usepackage[margin=1in]{geometry}
\usepackage{graphicx}
\usepackage[utf8]{inputenc}
\usepackage{CJKutf8}
\usepackage{longtable}
\usepackage{subcaption}
\usepackage{url}
\usepackage{verbatim}
\usepackage[normalem]{ulem}

\usepackage{hyperref}

\let\oldenumerate\enumerate
\renewcommand{\enumerate}{
  \oldenumerate
  \setlength{\itemsep}{1pt}
  \setlength{\parskip}{0pt}
  \setlength{\parsep}{0pt}
}
\let\olditemize\itemize
\renewcommand{\itemize}{
  \olditemize
  \setlength{\itemsep}{1pt}
  \setlength{\parskip}{0pt}
  \setlength{\parsep}{0pt}
}

\newcommand{\bbZ}{\mathbb{Z}}
\newcommand{\bbR}{\mathbb{R}}
\newcommand{\calI}{\mathcal{I}}
\newcommand{\calJ}{\mathcal{J}}
\newcommand{\calP}{\mathcal{P}}
\newcommand{\calQ}{\mathcal{Q}}
\newcommand{\calS}{\mathcal{S}}
\newcommand{\Tiling}{{\scshape Tiling}}
\newcommand{\FC}{{\scshape Flat Cover}}
\newcommand{\PE}{{\scshape 3-Precoloring Extension}}
\newcommand{\TS}{{\scshape Troublesome Sticker}}

\newcommand{\XTC}{{\scshape X3C}}

\newcommand{\Qpink}{{\color{Magenta}?\ }}
\newcommand{\half}{\frac{1}{2}}

\NewEnviron{problem}[2]{
  \begin{center}\fbox{\parbox{4.5in}{
        {\centering {\scshape #1} #2\par}
        \parskip=1ex
        \everypar{\hangindent=1em}
        \BODY
    }}
\end{center}}

\theoremstyle{definition}
\newtheorem{definition}{Definition}[section]

\newtheorem{lemma}[definition]{Lemma}
\newtheorem{theorem}[definition]{Theorem}
\newtheorem{corollary}[definition]{Corollary}
\newtheorem{remark}[definition]{Remark}

\newtheorem{question}{Question}

\title{Covering a Polyomino-Shaped Stain with Non-Overlapping Identical Stickers}
\author{Keigo Oka\thanks{Google. Work done in a personal capacity.}\\
  \href{mailto:ogiekako@gmail.com}{ogiekako@gmail.com}\and
  Naoki Inaba\thanks{Puzzle creator}\\
  \href{mailto:puzlab@gmail.com}{puzlab@gmail.com}\and
  Akira Iino\thanks{Nippon Hyoron Sha, Co., Ltd.}\\
\href{mailto:iino@nippyo.co.jp}{iino@nippyo.co.jp}}
\date{}

\newcommand{\datadir}{./}

\begin{document}

\maketitle

\begin{abstract}
  You find a stain on the wall and decide to cover it with non-overlapping stickers of a single identical shape (rotation and reflection are allowed). Is it possible to find a sticker shape that fails to cover the stain? In this paper, we consider this problem under polyomino constraints and complete the classification of always-coverable stain shapes (polyominoes). We provide proofs for the maximal always-coverable polyominoes and construct concrete counterexamples for the minimal not always-coverable ones, demonstrating that such cases exist even among hole-free polyominoes. This classification consequently yields an algorithm to determine the always-coverability of any given stain. We also show that the problem of determining whether a given sticker can cover a given stain is $\NP$-complete, even though exact cover is not demanded. This result extends to the 1D case where the connectivity requirement is removed. As an illustration of the problem complexity, for a specific hexomino (6-cell) stain, the smallest sticker found in our search that avoids covering it has, although not proven minimum, a bounding box of $325 \times 325$.
\end{abstract}

\section{Introduction}

Let $U$ be a set, $\calP, \calQ\subset 2^U$ be families of subsets of $U$, and let an equivalence relation $\sim$ be fixed on $\calP$. For a given $P\in \calP$ (sticker) and $Q\in\calQ$ (stain), if there are $P_1, \cdots, P_k\in\calP$ (where $k$ is a non-negative integer) such that each of them is equivalent to $P$, $P_i \cap P_j \neq \emptyset \Rightarrow i = j$, and $Q \subset P_1 \cup \cdots \cup P_k$, we say that $P$ \dotuline{flatly covers} $Q$. Hereafter, unless otherwise noted, ``\dotuline{cover}" will mean ``flatly cover". When $\calP,\sim,\calQ$ are fixed, the following problem,
which we name \TS{}, can be considered.

\begin{problem}{Troublesome Sticker}{$(\calP,\sim,\calQ)$}
  Input: $Q\in\calQ$

  Question: Is $Q$ \dotuline{always-coverable}? That is, is it true that for any $P\in\calP$, $P$ flatly covers $Q$?
\end{problem}

In this paper, (1) we provide a necessary and sufficient condition for a polyomino to be always-coverable. Surprisingly, no polyomino of size 7 or larger is always-coverable. For every polyomino that is always-coverable, we prove that it can be covered by any polyomino sticker, and for every polyomino that is not always-coverable, we provide a concrete sticker that cannot cover it. These were found through computational search and rigorous verification. (2) We prove $\NP$-completeness of \FC{} for polyominoes, which asks if a given stain can be covered by a given sticker, and $\NP$-completeness of a 1D variant without connectivity requirement.

There are various studies on the problem of covering a plane with one type of 2D shape (\Tiling). The discovery of shapes which only admit aperiodic tiling~\cite{smith2023aperiodic} is well known. Restricting the domain from arbitrary 2D shapes to polyominoes, the necessary and sufficient condition for one type of polyomino to cover a plane by translation only, and an efficient way to determine the condition, are known~\cite{Beauquier1991-zm, Gambini2007-lp}. The restriction on the number of types is essential because if we allow the use of 11 different types of polyominoes, the decision problem becomes undecidable~\cite{Ollinger2009-zh}. Even if the connectivity requirement is dropped, it is known that it is still decidable whether one type can cover a plane~\cite{bhattacharya2020periodicity}. If the region to be (exactly) covered is not a plane but a given polyomino, then the problem of tiling a given region with one type of tromino (whether of type L or I), allowing rotations, is known to be $\NP$-complete, $\sharp\P$-complete, and $\ComplexityFont{ASP}$-complete~\cite{moore2001hard, horiyama2017complexity}.

While \Tiling{} requires that the region be exactly covered, some studies discuss problems that allow the covering to exceed the region to be covered. One such problem is whether a set of arbitrary $k$ points can be covered by circular coins of the same size. This problem was popularized by Peter Winkler's column~\cite{Winkler2010-el} which discussed the problem for $k=10$ proposed by Naoki Inaba~\cite{inabaSuuri4}, and while it is known that there is an arrangement of $k=53$ points that cannot be covered~\cite{okayama2011covering}, and it is proven that there is no such arrangement for $k\le 12$~\cite{Winkler2010-el}, the problem is unsolved for intermediate values of $k$ to the best of our knowledge.

\begin{figure}[t]
  \centering
  \begin{minipage}[c]{0.48\textwidth}
    \centering
    \includegraphics[width=\linewidth]{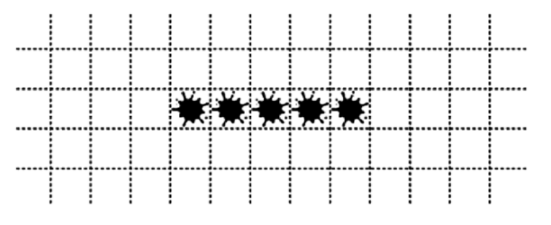}
    \caption{Pentomino~I ($1 \times 5$) shaped stain}
    \label{fig:ip}
  \end{minipage}
  \hfill
  \begin{minipage}[c]{0.48\textwidth}
    \centering
    \includegraphics[width=\linewidth]{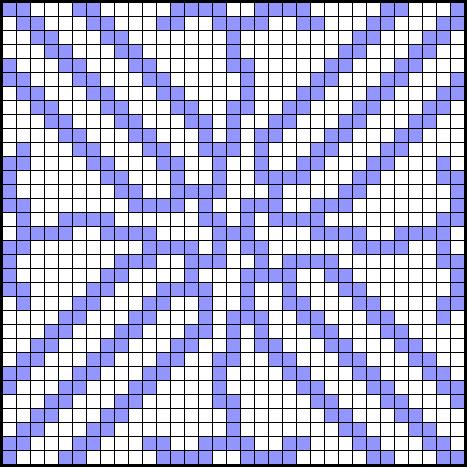}
    \caption{The smallest known polyomino ($33 \times 33$) that never covers pentomino~I}
    \label{fig:ipa}
  \end{minipage}
\end{figure}

Common to the problems we have seen so far is the requirement to cover a region with non-overlapping congruent shapes.
We refer to the target region as the \emph{stain} and the covering shape as the \emph{sticker}.
In the problem of covering points with coins, the shape of the stickers (disk of a given size) was already given, and the shape of the stain that cannot be covered by them was of our interest.
Our study is motivated by the idea of swapping the roles of the input: instead of fixing the sticker, we fix the stain. That is, we consider the problem of finding a sticker that cannot cover a given stain (\TS{}). We will see this in detail in Section~\ref{sec:polyomino_ts}, but we give here an example. If the stain is the pentomino~I (Figure~\ref{fig:ip}), is there a polyomino that cannot cover it (allowing rotations and reflections)? The answer is yes, and Figure~\ref{fig:ipa} shows the smallest such polyomino we know of~\cite{inabaSuuri7}. It is interesting that the complex shape emerges from a simple question. The main result we present in this paper is that for every polyomino-shaped stain, we have either shown a polyomino that cannot cover it, or given a proof that any polyomino can cover it. Polyominoes were chosen as our domain because they are a popular subject of study (popularized by the work of Solomon Golomb~\cite{golomb1996polyominoes} and Martin Gardner~\cite{gardner1960more}), and are easy to handle by computer programs thanks to their discrete nature. Although not treated in this paper, studying the problem on continuous objects will also be fun.

The last section presents open problems that arose during our research process.

\section{Preliminaries}\label{sec:preliminaries}

A non-empty finite set $P\subset\bbZ^2$ is called a \dotuline{polyomino} if and only if the interior of $P + \left[-\half,\half\right]^2$ is connected, where $P + \left[-\half,\half\right]^2 = \left\{a+(x,y) \mid a\in P, x,y\in \left[-\half,\half\right]\right\}$ is a subset of $\bbR^2$.

Let $P$ be a polyomino. A polyomino obtained by \dotuline{translating} $P$ is any polyomino that is equal to $P+d$ for some $d\in\bbZ^2$. A polyomino obtained by \dotuline{rotating} $P$ is any polyomino obtained by rotating $P$ around the origin by a multiple of 90 degrees. A polyomino obtained by \dotuline{reflecting} $P$ is any polyomino obtained by reflecting $P$ about the $x$ or $y$ axis. We consider the positive direction of the $x$-axis as the \dotuline{right direction} and the positive direction of the $y$-axis as the \dotuline{up direction}.

$P+\left[-\half,\half\right]^2$, which is a subset of $\bbR^2$, is commonly called \dotuline{the animal of $P$}. Since there is a one-to-one correspondence between a polyomino and its animal, it is common practice to draw its animal when illustrating a polyomino. If specific coordinates are omitted in the illustration, it means the origin can be anywhere. Sometimes $p\in\bbZ^2$ is illustrated by the animal of $\{p\}$ and is called a \dotuline{square}. \dotuline{The upper right point of a square $p$} is the upper rightmost point of the animal of $\{p\}$, and the same applies to upper left, lower left, and lower right. A square contained in the polyomino under consideration is sometimes called its \dotuline{cell}.

\dotuline{The bounding box of $P$} is the smallest rectangle with both sides parallel to the $x$ and $y$ axes that includes the animal of $P$. \dotuline{The top side of $P$} is the top side of the bounding box of $P$. Note that reflecting $P$ about its top side yields another polyomino, which has no element in common with $P$. The bottom, left, and right sides of $P$ are defined in the same way. \dotuline{The width (resp. height) of $P$} is the width (resp. height) of the bounding box of $P$. For $p\in P$, \dotuline{the row index of $p$ in $P$} is $\half$ plus the difference between the $y$ coordinate of the top side of $P$ and the $y$ coordinate of $p$ (the minimum possible row index is 1). \dotuline{The $i$-th row of $P$} is the set of elements in $P$ whose row index in $P$ is $i$.

Two polyominoes are defined to \dotuline{overlap} if their intersection is not empty. The operation of taking the union of two polyominoes that do not overlap is specifically denoted by $\oplus$.

In Section~\ref{sec:polyomino_ts} and Appendix~\ref{sec:coverable}, we use the term \dotuline{include} in a specific sense: for a set to \dotuline{include} a polyomino $Q$ means that the set contains a congruent copy of $Q$ (i.e., at least one of the polyominoes obtained by arbitrarily translating, rotating, or reflecting $Q$).

\section{Troublesome Sticker on Polyominoes}\label{sec:polyomino_ts}

Let $\calP$ be the set of all polyominoes. We write $\sim_E$ to denote the equivalence relation that equates two polyominoes that can become the same after applying arbitrary translations, rotations, and reflections. In this section we show the following theorem.

\begin{theorem}\label{thm:p_ts}
  There exists an algorithm that solves \TS~$(\calP,\sim_E,\calP)$ in $O(1)$ time.
\end{theorem}

For a polyomino $Q$, we say that $Q$ is \dotuline{always-coverable} if any polyomino can cover $Q$. Theorem~\ref{thm:p_ts} in other words means we can determine in $O(1)$ time whether a given polyomino is always-coverable or not. It is clear that the following remarks hold.

\begin{remark}\label{rem:ac}
  \leavevmode
  \begin{itemize}
    \item If $Q$ is always-coverable, any polyomino included in $Q$ is also always-coverable.
    \item If $Q$ is not always-coverable, any polyomino that includes $Q$ is also not always-coverable.
  \end{itemize}
\end{remark}

\begin{remark}\label{rem:dfs}
  Determining if a polyomino $P$ of size $n$ can cover a polyomino $Q$ of size $k$ is solvable in $O^*(n^k)$ time via depth-first search.
\end{remark}

\subsection{Not always-coverable Polyominoes}

From Remark~\ref{rem:dfs}, determining whether a specific sticker flatly covers a specific small stain is computationally tractable. Thus, verifying a potential counterexample (a sticker that fails to cover the stain) is straightforward using a computer. We computationally found and verified counterexamples for all polyominoes in the set $\calI$ (Figure~\ref{fig:IandJ}). The exact shapes of these counterexamples are detailed in Appendix~\ref{sec:uncoverable}, establishing the following theorem.

\begin{theorem}\label{thm:uncoverable}
  Polyominoes in $\calI$ are not always-coverable.
\end{theorem}

Every polyomino with 7 cells includes at least one polyomino in $\calI$. (This can be easily shown by a computational search). From this and Remark~\ref{rem:ac}, we obtain the following.

\begin{corollary}\label{col:ac7}
  Polyominoes with 7 or more cells are not always-coverable.
\end{corollary}

\begin{figure}
  \centering
  \includegraphics[scale=0.45]{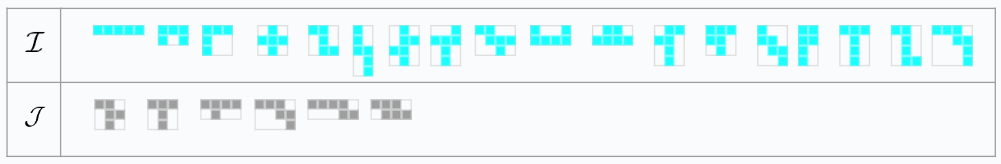}
  \captionof{figure}{Not always-coverable polyominoes $\calI$ and always-coverable polyominoes $\calJ$}
  \label{fig:IandJ}
\end{figure}

The algorithm used to search for counterexamples is based on the simulated annealing method~\cite{van1987simulated}. For computational efficiency, the shape of the stickers examined by the algorithm was restricted to rotationally and reflectionally symmetric trees (only the central part may be asymmetric so that the candidate itself does not include the stain). Neighborhoods for the annealing method include adding/removing a cell, random state flipping in a $2\times 2$ and $3\times 3$ region, and swapping the states of two random squares. The penalty is basically the number of ways to cover the stain with two stickers, but it is made higher if the stain can be covered with the cells near the side of the sticker or the outermost diagonal line that intersects the sticker. Also, for each way to cover the stain with two stickers, the smaller the number of ways to prevent the covering by adding a cell, the higher the penalty. The resulting behavior we see is that, over time, the average position of the cells contributing to the covering moves away from the side or diagonal line towards the center, and eventually a counterexample is found. We applied speed-up techniques such as delta-updating the penalty, and limiting the depth of interference of two stickers at the expense of the accuracy of the coverability check, and after finding candidates for counterexamples, these were fully checked using a naive algorithm. The program used is available for download at \url{https://drive.google.com/file/d/13V3XtuAb14-E6APNlQB5YRVExtIwddTd}.\footnote{SHA-256 hash: 7b9ce75eef1d973f6ab97227f7101e9984b1d4b7c5371b30b732a67c78377d18}

\subsection{Always-coverable Polyominoes}

For each polyomino in $\calJ$ in Figure~\ref{fig:IandJ}, we prove that it is always-coverable by contradiction. That is, we assume that there exists a polyomino (counterexample) which cannot cover the stain and derive a contradiction. In this section we give a proof that pentomino~Y (Figure~\ref{fig:pentominoY}) and pentomino~T (Figure~\ref{fig:pentominoT}) are always-coverable, and in Appendix~\ref{sec:coverable}, we give a proof for other polyominoes. From the following remark, we can assume that the counterexamples we assume in our proof are sufficiently large.

\begin{remark}
  Suppose that a polyomino $Q$ is not always-coverable. Then, for any positive integer $n$, there exists a polyomino whose height, width, and number of cells are all greater than or equal to $n$ that cannot cover $Q$.
\end{remark}

\begin{proof}
  Given a sticker $P_0$ that fails to cover $Q$, define $P'_i$ as the union of $P_i$ and the reflection of $P_i$ about its right side, and $P_{i+1}$ as the union of $P'_i$ and the reflection of $P'_i$ about its top side, then $P_{\lceil \log_2n \rceil}$ is a polyomino whose height, width, and number of cells are all $\ge n$ that cannot cover $Q$.
\end{proof}

\begin{figure}
  \centering
  \begin{minipage}{.45\textwidth}
    \centering
    \includegraphics[scale=0.3]{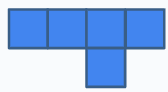}
    \captionof{figure}{\\Pentomino ``Y"}
    \label{fig:pentominoY}
  \end{minipage}
  \begin{minipage}{.45\textwidth}
    \centering
    \includegraphics[scale=0.3]{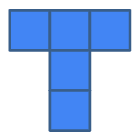}
    \captionof{figure}{\\Pentomino ``T"}
    \label{fig:pentominoT}
  \end{minipage}
\end{figure}

\begin{figure}
  \centering
  \begin{minipage}{.29\textwidth}
    \centering
    \includegraphics[width=.75\linewidth]{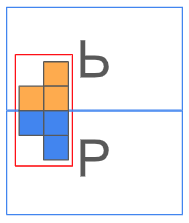}
    \captionof{figure}{Consecutive cells in the first row create a contradiction}
    \label{fig:proof_5y_1}
  \end{minipage}
  \hspace{0.04\textwidth}
  \begin{minipage}{.29\textwidth}
    \centering
    \includegraphics[width=0.95\linewidth]{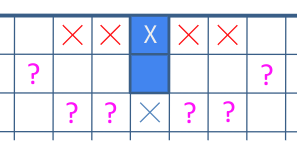}
    \captionof{figure}{A contradiction occurs if $X$ is connected to any \Qpink}
    \label{fig:proof_5y_2}
  \end{minipage}
  \hspace{0.04\textwidth}
  \begin{minipage}{.29\textwidth}
    \centering
    \includegraphics[width=0.72\linewidth]{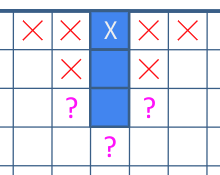}
    \captionof{figure}{A contradiction occurs if $P$ contains any \Qpink}
    \label{fig:proof_5y_3}
  \end{minipage}
\end{figure}

\begin{lemma}\label{lem:cover5Y}
  Pentomino~Y is always-coverable.
\end{lemma}

\begin{proof}
  Let $Q$ be a pentomino~Y and assume that $Q$ cannot be covered by a sufficiently large polyomino $P$. In particular, we assume the width and height of $P$ are both greater than 6. Let $P_{\mathrm{rev}}$ be the reflection of $P$ about its top side. There are no two consecutive cells (cells with a distance of 1) in the first row of $P$, because otherwise there should be two consecutive cells in the first row one of which is adjacent to a cell in the second row, and thus $P_{\mathrm{rev}}\oplus P$ includes $Q$ (Figure~\ref{fig:proof_5y_1}), which is a contradiction. Hence, the polyomino obtained by translating $P_{\mathrm{rev}}$ by one square to the left (resp. right) and then by one square down, which we call $P_{\mathrm{L}}$ (resp. $P_{\mathrm{R}}$), does not overlap $P$. This means that there are also no two cells in the first row of $P$ with a distance of 2 (otherwise, $P_{\mathrm{L}}\oplus P$ would include $Q$). Hence, the translation of $P_{\mathrm{L}}$ (resp. $P_{\mathrm{R}}$) by one square to the left (resp. right), which we call $P_{\mathrm{L2}}$ (resp. $P_{\mathrm{R2}}$), also does not overlap $P$. We now turn our attention to an arbitrary cell $X$ in the first row of $P$. We can say that there is a cell on the square two below $X$ (otherwise, there would be a cell at one of the positions marked by \Qpink in Figure~\ref{fig:proof_5y_2} that connects to $X$ without passing through any other square marked with \Qpink, but in any case, one of $P_{\mathrm{L2}}\oplus P, P_{\mathrm{L}}\oplus P, P_{\mathrm{R}}\oplus P, P_{\mathrm{R2}}\oplus P$ would include $Q$, which is a contradiction). Therefore, $P$ does not contain the squares lower left and lower right of $X$ (otherwise, $P_{\mathrm{rev}}\oplus P$ would include $Q$). Hence, the translation of $P_{\mathrm{L}}$ (resp. $P_{\mathrm{R}}$) by one square down, which we call $P_{\mathrm{LD}}$ (resp. $P_{\mathrm{LR}}$), also does not overlap $P$. There is a cell in one of the positions marked by \Qpink in Figure~\ref{fig:proof_5y_3}, but in any case it leads to a contradiction because $P_{\mathrm{LD}}\oplus P$ or $P_{\mathrm{RD}}\oplus P$ would include $Q$.
\end{proof}

\begin{lemma}\label{lem:cover5T}
  Pentomino~T (Figure~\ref{fig:pentominoT}) is always-coverable.
\end{lemma}

\begin{figure}
  \centering
  \begin{minipage}{.40\textwidth}
    \centering
    \includegraphics[width=.9\linewidth]{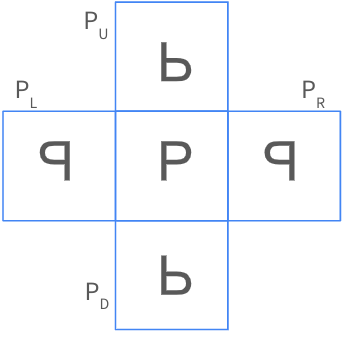}
    \captionof{figure}{\\Definitions of $P_\mathrm{U}, P_\mathrm{D}, P_\mathrm{L}, P_\mathrm{R}$}
    \label{fig:5Tsetup1}
  \end{minipage}
  \begin{minipage}{.40\textwidth}
    \centering
    \includegraphics[width=.9\linewidth]{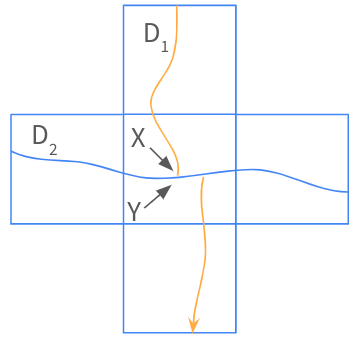}
    \captionof{figure}{\\Definitions of $D_1, D_2, X, Y$}
    \label{fig:5Tsetup2}
  \end{minipage}
\end{figure}

\begin{figure}
  \centering
  \begin{minipage}{.25\textwidth}
    \centering
    \includegraphics[width=.9\linewidth]{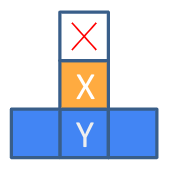}
    \captionof{figure}{}
    \label{fig:5Tcontra1}
  \end{minipage}
  \begin{minipage}{.25\textwidth}
    \centering
    \includegraphics[width=.9\linewidth]{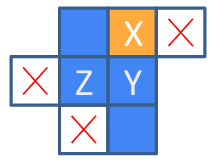}
    \captionof{figure}{}
    \label{fig:5Tcontra2}
  \end{minipage}
  \begin{minipage}{.25\textwidth}
    \centering
    \includegraphics[width=.9\linewidth]{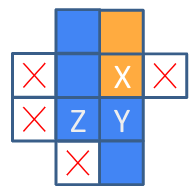}
    \captionof{figure}{}
    \label{fig:5Tcontra3}
  \end{minipage}
\end{figure}

\begin{proof}
  Let $Q$ be a pentomino~T and assume that $Q$ cannot be covered by a sufficiently large polyomino $P$. Consider the polyomino $P_\mathrm{U}\oplus P_\mathrm{D}\oplus P_\mathrm{L}\oplus P_\mathrm{R}\oplus P$ where $P_\mathrm{U}, P_\mathrm{D}, P_\mathrm{L}, P_\mathrm{R}$ are the polyominoes obtained by reflecting $P$ about its top, bottom, left, and right side respectively (Figure~\ref{fig:5Tsetup1}). From the assumption this does not include $Q$. Consider two cells with a distance of 1 as adjacent, and look at the oriented shortest path $D_1$ from any top-side cell to a bottom-side cell in $P_\mathrm{U}\oplus P_\mathrm{D}\oplus P$ and the shortest path $D_2$ from any left-side cell to a right-side cell in $P_\mathrm{L}\oplus P_\mathrm{R}\oplus P$. Note that $D_1$ and $D_2$ must have an intersection in $P$. Since $D_1\cup D_2\subset P_\mathrm{U}\oplus P_\mathrm{D}\oplus P_\mathrm{L}\oplus P_\mathrm{R}\oplus P$, $D_1\cup D_2$ does not include $Q$. Let $X$ be the first cell in $D_1$ that is adjacent to any cell in $D_2$, and $Y$ be any cell in $D_2$ adjacent to $X$, and without loss of generality, assume that $X$ is the cell above $Y$ (Figure~\ref{fig:5Tsetup2}). From the way $X$ is taken, $X\notin D_2$. Since $Y$ is a cell on the shortest path $D_2$, $D_2$ contains exactly two cells adjacent to $Y$. It is not possible for $D_2$ to contain both the left and right squares of $Y$ (otherwise since $D_1\cup D_2$ does not include $Q$, $D_1$ would not contain the square above $X$, and then the cell before $X$ in $D_1$ must be adjacent to a cell in $D_2$, which is a contradiction (Figure~\ref{fig:5Tcontra1})). Hence, $D_2$ contains the square below $Y$. It also contains either the square to the left or to the right of $Y$. We assume that it contains the left of $Y$ and call it $Z$ (if it contains the right of $Y$, the same argument holds). From its shortestness, $D_2$ does not contain the square under $Z$. In order not to cover $Q$, it also does not contain the left of $Z$ and thus contains the square above $Z$. In order not to cover $Q$, the right of $X$ is not contained in $D_1$ (Figure~\ref{fig:5Tcontra2}). From the way $X$ is taken, $D_1$ contains the square above $X$ as the cell before $X$ in $D_1$. In order not to cover $Q$, $D_2$ does not contain the square upper left of $Z$, and thus contains the square two above $Z$. Then the cell above $X$ becomes adjacent to a cell in $D_2$, which is a contradiction (Figure~\ref{fig:5Tcontra3}).
\end{proof}

\begin{theorem}\label{thm:coverable}
  Polyominoes in $\calJ$ are always-coverable.
\end{theorem}

\begin{proof}
  From Lemma~\ref{lem:cover5Y}~and~\ref{lem:cover5T}, pentomino~Y~and~T are always-coverable. Any other polyominoes in $\calJ$ are proven to be always-coverable in Appendix~\ref{sec:coverable}.
\end{proof}

\subsection{Proof of Theorem~\ref{thm:p_ts}}

\begin{proof}
  Every polyomino either includes a polyomino in $\calI$ or is included in a polyomino in $\calJ$. For polyominoes with 7 or more cells, this has already been shown by Corollary~\ref{col:ac7}. For polyominoes with fewer than 7 cells, this is easily confirmed by a computerized exhaustive search.

  Let us give an algorithm that determines whether a given polyomino $Q$ is always-coverable in $O(1)$ time. From Corollary~\ref{col:ac7}, if the number of cells in $Q$ is greater than or equal to 7, we can immediately answer that it is not always-coverable, so we can assume that the number of cells is less than 7. $Q$ either includes a polyomino in $\calI$ or is included in a polyomino in $\calJ$, and we can determine which is the case in $O(1)$~time. If the former is the case, $Q$ is not always-coverable; if the latter is the case, $Q$ is always-coverable.
\end{proof}

It is also possible, in linear time in the input size, to output a polyomino that cannot cover a given polyomino $Q$ if it is not always-coverable. We can find a polyomino in $\calI$ that $Q$ includes and output the known polyomino that cannot cover it. Since the size of $\calI$ is $O(1)$, this algorithm works in linear time.

Finally, it is notable that all the counterexamples we have given in Appendix~\ref{sec:uncoverable} have no holes, i.e., all the animals induced from them are simply connected. This means that even if we limit the domain to simply connected polyominoes, we still have the same result as Theorem~\ref{thm:p_ts} via the same proof:

\begin{theorem}
  Let $\calP'$ be the set of all simply connected polyominoes, and $\sim_E$ be the equivalence relation under translation, rotation, and reflection.
  \TS$(\calP', \sim_E, \calP')$ is solvable in $O(1)$ time, and if the stain is not always-coverable, a sticker that cannot cover it can be found in linear time.
\end{theorem}

\section{NP-hardness of Flat Cover}
\label{sec:hardness}

For any \TS$(\calP, \sim, \calQ)$, the following decision problem, which we name \FC{}, can be considered.

\begin{problem}{Flat Cover}{$(\calP, \sim,\calQ)$}
  Input: $P\in\calP, Q\in\calQ$

  Question: Can $P$ flatly cover $Q$?
\end{problem}

In this section, we show that \FC{} for polyominoes is $\NP$-hard.

\begin{theorem}\label{thm:bc_hardness}
  Let $\calP$ be the set of all polyominoes, and $\sim_E$ be the equivalence relation under translation, rotation, and reflection. Then,
  \FC$(\calP, \sim_E, \calP)$ is $\NP$-complete.
\end{theorem}

We show this by a reduction from \PE{} on induced subgrids.

\begin{theorem}[\PE{} on induced subgrids \cite{demange2013coloring}]\label{thm:precol}
  Let $G = (V, E)$ be an induced subgrid, i.e., an induced subgraph of the infinite integer grid $\bbZ^2$ where $E = \{ \{u, v\} \subset V \mid \|u - v\|_1 = 1 \}$.
  Given a subset $S \subset V$ and a partial coloring $\varphi : S \to \{1, 2, 3\}$ (a precoloring), the decision problem of whether there exists a proper 3-coloring $c : V \to \{1, 2, 3\}$ such that $c(v) = \varphi(v)$ for all $v \in S$ and $c(u) \neq c(v)$ for all $\{u, v\} \in E$ is $\NP$-complete.
\end{theorem}

\begin{figure}[htbp]
  \centering
  \begin{minipage}[c]{0.45\textwidth}
    \centering
    \begin{minipage}{10em}
      \small
\begin{verbatim}
0000001000
0110001110
0111111110
0011111100
0011111100
1111111111
0011111100
0111111110
0100001110
0000001000
\end{verbatim}
    \end{minipage}
    \caption{The gadget $P$. Each `1' represents a cell.}
    \label{fig:gadget_p}
  \end{minipage}
  \hfill
  \begin{minipage}[c]{0.45\textwidth}
    \centering
    \begin{minipage}{8em}
      \small
\begin{verbatim}
11000011
01111111
01111110
01111110
01111110
01111110
11111110
10000011
\end{verbatim}
    \end{minipage}
    \caption{The gadget $Q_0$, defined as $P_1 \cap P_2 \cap P_3$.}
    \label{fig:gadget_q0}
  \end{minipage}
\end{figure}

\begin{proof}[Proof of Theorem \ref{thm:bc_hardness}]
  Since it is clear that this problem belongs to $\NP$, we focus on showing $\NP$-hardness.
  We reduce \PE{} on induced subgrids to \FC{}.
  Let an instance of \PE{} be given by an induced subgrid $G=(V, E)$ and a partial coloring $\varphi : S \to \{1, 2, 3\}$. W.l.o.g. we can assume $G$ is connected.
  We construct a sticker $P$ and a stain $Q$ in polynomial time.

  Let $P$ be the $10 \times 10$ polyomino defined in Figure \ref{fig:gadget_p}.
  Let $P_1$, $P_2$, $P_3$ denote $P$, $\mathrm{rot}(P, 90^\circ)$, and $\mathrm{rot}(P, -90^\circ)$ respectively, corresponding to colors $1, 2, 3$.

  For each vertex $v \in V$ of the induced subgrid, we associate a position in the plane at coordinates $8v$ (scaling by factor 8).
  The sticker $P$ is designed to have the following properties that we have computationally verified with an exhaustive search.

  \begin{enumerate}
    \item There exists an $8 \times 8$ polyomino $Q_0$ derived from $P$ (Figure~\ref{fig:gadget_q0}) such that $Q_0$ has exactly three minimal flat covers (covers without any sticker disjoint from $Q_0$) by $P$, where either $P_1$, $P_2$, or $P_3$ is placed locating its center at the center of $Q_0$.
    \item Each $P_i$ ($i \in \{1, 2, 3\}$) has exactly one minimal flat cover by $P$, where $P_i$ (sticker) is placed locating its center at the center of $P_i$ (stain).
    \item When two stickers are placed at adjacent grid positions (distance 8 apart) each having the orientation of either $P_1, P_2,$ or $P_3$, the two stickers don't overlap if and only if they have distinct colors. Specifically, the protrusions of two stickers with the same orientation placed at distance 8 overlap, preventing such placement.
  \end{enumerate}

  For reduction, we place the following for each $v \in V$: If $v \notin S$ (uncolored), place $Q_0$ at position $8v$. If $v \in S$ with $\varphi(v) = i$, place $P_i$ at position $8v-(1,1)$ (so that its $8\times 8$ core aligns with $8v$). The total stain $Q$ is the union of all these regions.

  A flat cover of $Q$ by $P$ exists if and only if there exists a valid 3-coloring extension of $\varphi$ to $G$. The correctness follows directly from the three properties of $P$.

  Since $G$ is connected, the construction and $Q_0$'s shape ensure $Q$ is connected. Thus, both $P$ and $Q$ are polyominoes, concluding the proof.
\end{proof}

\section{NP-completeness in One Dimension}
\label{sec:1d_hardness}

In the previous section, we showed that \FC{} is $\NP$-complete for 2D polyominoes. While the problem is computationally trivial for connected 1D segments (as a single segment can trivially cover the line), generalizing this base case in two distinct directions---either by increasing the dimension to 2D (as shown previously) or by removing the connectivity constraint in 1D (as shown in this section)---leads to $\NP$-completeness.

\begin{theorem}
  \label{thm:1d_disconnected}
  Let $\calP$ be the set of all finite subsets of $\bbZ$, and let $\calQ = \{ [0, L) \cap \bbZ \mid L \in \bbZ^+ \}$ be the set of solid segments starting at 0. Let $\sim_T$ be the equivalence relation under translation.
  The problem \FC$(\calP, \sim_T, \calQ)$, which asks if a segment of length $L$ can be covered by non-overlapping translations of a single template $T \in \calP$, is $\NP$-complete.
\end{theorem}

Membership in $\NP$ is trivial. We prove $\NP$-hardness of this problem by a reduction from \XTC{}, which is an $\NP$-complete problem~\cite{DBLP:books/fm/GareyJ79} defined as follows: Given a universe $U$ of size $3q$ and a collection $\calS = \{S_1, \dots, S_r\}$ of 3-element subsets of $U$, does there exist a subcollection $\calS' \subset \calS$ such that every element of $U$ belongs to exactly one subset in $\calS'$?

The idea of the reduction is as follows.
We introduce two types of gadgets: one ``frame''  gadget $G_0$ and $r$ ``set'' gadgets $G_1, \dots, G_r$, and place them in a template $T$ with strategic distances based on a \emph{Golomb ruler}~\cite{Sidon1932,babcock1953intermodulation} (a set of integers with distinct pairwise differences).
The construction enforces that if there exists a flat cover of the target segment by $T$, it must use $q+1$ copies of $T$, where the segment cover is contributed by exactly one frame gadget and $q$ set gadgets, which represents a solution to the \XTC{} instance. The Golomb ruler property ensures that the templates do not overlap with each other, and the shape of the gadgets which we detail below ensures different types of solutions are impossible.

\begin{proof}
  Given an instance $(U = \{0, \dots, 3q-1\}, \calS = \{S_1, \dots, S_r\})$ of \XTC{}, we construct a target interval $[0, L)$ and a single template $T \subset \bbZ$.

  We let $N = 10(3q + r + 1)$, which we call \emph{element size}, and $L = 2N^2 + 3qN$ be the size of the target interval $Q=\{0, \dots, L-1\}$. We also define $W = L + 10(r+1)$, which we call \emph{gadget size}, because each gadget is represented as a binary string of length $W$.

  The template $T$ is constructed as a union of the gadgets $G_0, G_1, \dots, G_r$, placed at relative positions $0, 2a_1W, 2a_2W, \dots, 2a_rW$, where $\{a_0 = 0, a_1, \dots, a_r\}$ forms a Golomb ruler, i.e., all differences $a_i - a_j$ ($i \neq j$) are non-zero and distinct. While finding an optimal Golomb ruler is hard, a ruler with $a_r \in O(r^2)$ can be explicitly constructed in polynomial time~\cite{erdos1941problem, drakakis2009review}.

  Each gadget consists of a \emph{left stopper}, a \emph{main body}, and a \emph{right stopper}, whose lengths are $5(r+1)$, $L$, and $5(r+1)$ respectively. Each right stopper is a reversed binary string of the left stopper.
  We denote a run of $k$ ones as $1^k$ and $k$ zeros as $0^k$.

  The frame gadget $G_0$ is defined as follows.

  \[
    G_0 = \underbrace{0^{5r}11110}_{\text{left stopper}} \circ \underbrace{1^{N^2} \, 0^{3qN} \, 1^{N^2}}_{\text{main body}} \circ \underbrace{011110^{5r}}_{\text{right stopper}}
  \]

  For each $i \in \{1, \dots, r\}$, $G_i$ (set gadget) is defined as follows.
  \[
    G_i = 0^{5(r-i)}110110^{5i} \circ 0^{N^2} \mathrm{code}(S_i) 0^{N^2} \circ 0^{5i}110110^{5(r-i)}
  \] where $\mathrm{code}(S_i) = \bigcup_{j \in S_i} 0^{jN}1^{N}0^{(3q-1-j)N}$ is a binary string of length $3qN$ encoding the set $S_i$.

  The template is defined as $T = \bigcup_{i \in \{0, \dots, r\}} 0^{2a_iW}G_i$. The total length of $T$ is polynomial in the input size since $a_r \in O(r^2)$.
  By construction, the distance between different gadgets $G_i$ and $G_j$ is $2(a_i-a_j)W$, and placing two copies of $T$ with relative shift of $2(a_i-a_j)W$ results in a collision if and only if $i,j >0$ and $S_i\cap S_j\neq\emptyset$.

  Note the following properties:
  \begin{itemize}
    \item Due to the block of 1's in the main body of the frame gadget, placing two copies of $T$ with relative shift of less than $L$ results in a collision.
    \item In each gadget, every 1 has another 1 to its immediate left or right, meaning filling a single 0 between 1's without collision is impossible.
  \end{itemize}

  \paragraph{Correctness}
  ($\Rightarrow$) Suppose $S_{k_1}, \dots, S_{k_q} \subset \calS$ is an exact cover. We place one copy of $T$ so that the main body of its $G_0$ aligns with $Q$. For each $i \in \{1, \dots, q\}$, we place one copy of $T$ so that the main body of its $G_{k_i}$ aligns with $Q$. This forms a valid cover of $Q$ by $T$, because the solution is an exact cover, so the selected sets are disjoint, ensuring that the placements do not collide, and $Q$ is clearly fully covered.

  ($\Leftarrow$) Suppose $Q$ is covered by a set of non-overlapping $m+1$ copies of $T$, and let $G_{k_0}, \dots, G_{k_m}$ be all the gadgets contributing to the cover of $Q$. We can assume $m+1$ gadgets here because due to the distances between gadgets, at least one gadget per a copy of $T$ can intersect with $Q$.

  We first prove that $\{k_i\}$ must contain $0$ (frame gadget), assuming otherwise, i.e., $Q$ is fully covered by set gadgets. Since each set gadget has exactly $8 + 3N < 4N$ 1's, and $Q$ has more than $2N^2$ 1's, the number of used gadgets must be at least $N/2$, which is larger than $5r$, indicating that there exists an index (say $j$) contained in $\{k_i\}$ more than 4 times. Because of the length of $Q$ (which is $L$) and $G_j$ (which is $W < 2L$), simple arithmetic shows that there exist two copies of $T$ with relative shift less than $L$. However, this is impossible because it produces a collision on the frame gadgets.

  Let's consider how $G_0$ covers $Q$. Since its main body and left/right stoppers are separated by a single 0 between 1's, the intersection of $G_0$ and $Q$ must either entirely lie on the main body, or its left or right stopper. If we suppose $Q$ never intersects with $G_0$'s main body, since the stoppers have only four 1's, we can easily lead to a contradiction with the same argument as we used in the previous paragraph. Therefore, there exists a copy of $T$ covering $Q$ precisely with $G_0$'s main body. It is easy to see no other copy of $T$ can cover $Q$ with $G_0$ without collision, so we can set $k_0 = 0$ and $k_i > 0$ for $i \in \{1, \dots, m\}$.

  Now that $G_0$'s main body is placed at $Q$, the remaining central blank of length $3qN$ should be filled by $m$ set gadgets. For each of them, it is never possible for its left or right stopper to intersect with the blank, because of the pattern 11011 having a 0 impossible to fill without collision. That means the blank is filled by 1's in the main body of each set gadget, having the consecutive length that is a multiple of $N$, meaning the alignment of multiple of $N$ must hold. That is, a possible placement of a set gadget is placing its main body to precisely cover $Q$ or with a shift of $sN$ with an integer $s$ with $|s| < 3q$. However, choosing $s\neq 0$ is impossible as it causes its left or right stopper to collide with the 1's in the already-placed $G_0$'s main body. Thus, $m=q$ set gadgets' main bodies must align with $Q$, completing the proof.
\end{proof}

\section{Open Problems}

We have shown that \TS{} is decidable for polyominoes, but what if the domain is not polyominoes? The following problem is, to the best of our knowledge, unsolved even for $d=1$.

\begin{question}\label{que:ts}
  Let $d$ be a positive integer and $\calP$ be the set of all finite subsets of $\bbZ^d$. Define $\sim_T$ on $\calP$ so that for $A,B\in \calP$, $A\sim_T B\iff\exists x\in \bbZ^d, A=B+x$. Is \TS$(\calP,\sim_T, \calP)$ decidable?
\end{question}

As we have proven, \FC{} is $\NP$-complete for general polyominoes. What would happen if we limited the height of the stain?

\begin{question}
  Let $\calP$ be the set of polyominoes and $\sim_E$ be the equivalence relation under translation, rotation, and reflection. Let $\calQ$ be the set of all polyominoes with a height of 1. Is \FC$(\calP, \sim_E, \calQ)$ $\NP$-complete?
\end{question}

In \TS{}, the stain was given as input, but what would be the computational complexity of the problem of, conversely, finding a stain that cannot be covered by a given sticker?

\begin{question}
  Is the problem of finding a polyomino with the minimum number of cells that a given polyomino cannot flatly cover $\NP$-hard?
\end{question}

\bibliography{main}

\newpage

\section*{Appendix}
\appendix

\section{Complete Proof of Theorem~\ref{thm:coverable}}\label{sec:coverable}

\begin{proof}
  From Lemma~\ref{lem:cover5Y},~\ref{lem:cover5T} and the following Lemma~\ref{lem:cover5F},~\ref{lem:cover63},~\ref{lem:cover6N},~\ref{lem:cover6S}, every polyomino in $\calJ$ is always-coverable and thus Theorem~\ref{thm:coverable} holds.
\end{proof}

\begin{figure}[ht]
  \centering
  \begin{minipage}{.17\textwidth}
    \centering
    \includegraphics[width=.75\linewidth]{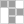}
    \captionof{figure}{\\Pentomino ``F"}
    \label{fig:stain5F}
  \end{minipage}
  \begin{minipage}{.18\textwidth}
    \centering
    \includegraphics[width=.9\linewidth]{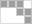}
    \captionof{figure}{\\Hexomino ``3"}
    \label{fig:stain63}
  \end{minipage}
  \begin{minipage}{.20\textwidth}
    \centering
    \includegraphics[width=.9\linewidth]{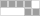}
    \captionof{figure}{\\Hexomino ``N"}
    \label{fig:stain6N}
  \end{minipage}
  \begin{minipage}{.18\textwidth}
    \centering
    \includegraphics[width=.9\linewidth]{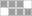}
    \captionof{figure}{\\Hexomino ``S"}
    \label{fig:stain6S}
  \end{minipage}
\end{figure}

\begin{lemma}\label{lem:cover5F}
  Pentomino~F (Figure~\ref{fig:stain5F}) is always-coverable.
\end{lemma}

\begin{figure}[ht]
  \centering
  \begin{minipage}{.35\textwidth}
    \centering
    \includegraphics[width=.9\linewidth]{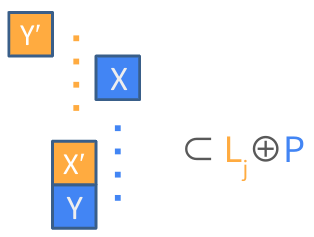}
    \captionof{figure}{}
    \label{fig:5Fsetup1}
  \end{minipage}
  \begin{minipage}{.20\textwidth}
    \centering
    \includegraphics[width=.9\linewidth]{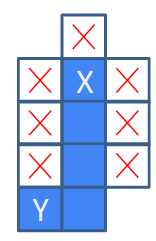}
    \captionof{figure}{}
    \label{fig:5Fcontra1}
  \end{minipage}
  \begin{minipage}{.20\textwidth}
    \centering
    \includegraphics[width=.9\linewidth]{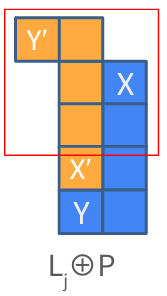}
    \captionof{figure}{}
    \label{fig:5Fcontra2}
  \end{minipage}
  \begin{minipage}{.20\textwidth}
    \centering
    \includegraphics[width=.9\linewidth]{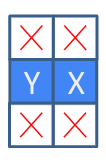}
    \captionof{figure}{}
    \label{fig:5Fcontra3}
  \end{minipage}
\end{figure}

\begin{proof}
  Let $Q$ be a pentomino~F and assume that $Q$ cannot be covered by a sufficiently large polyomino $P$. Let $L_0$ (resp. $R_0$) be the polyomino obtained by reflecting $P$ about its top side and shifting it to the left (resp. right) by one square, and let $L_i, R_i$ be polyominoes obtained by shifting $L_0, R_0$ down by $i$ squares, respectively. Take the smallest $j$ such that $L_{j+1}\cap P$ is non-empty and the smallest $k$ such that $R_{k+1}\cap P$ is non-empty. From symmetry, we may assume that $j\le k$. That is, none of $L_0, \cdots L_j, R_0, \cdots, R_j$ overlaps $P$. From the fact that $L_{j+1}$ overlaps $P$, we can take $X'\in L_j$ and $Y\in P$ such that $Y$ is one square below $X'$. Let $X$ be the cell in $P$ corresponding to $X'$ (Figure~\ref{fig:5Fsetup1}). In the following, we divide cases according to the difference between the row index of $X$ and the row index of $Y$ in $P$.

  \begin{itemize}
    \item If the row index of $X$ is smaller than the row index of $Y$, $P$ does not contain the squares marked by $\times$ in Figure~\ref{fig:5Fcontra1}, since $L_j$ does not contain the square below $X'$, and $L_0, \cdots, L_j, R_0, \cdots, R_j$ do not overlap $P$. Therefore, by connectedness, $P$ contains all the squares below $X$ up to the square to the right of $Y$. Then $L_j\oplus P$ contains $Q$, which is a contradiction (Figure~\ref{fig:5Fcontra2}).
    \item If the row index of $X$ is greater than the row index of $Y$, by reversing up and down and reversing the roles of $L_j$ and $P$, a contradiction is derived in the same way.
    \item If the row index of $X$ equals the row index of $Y$, $P$ does not contain any squares above or below $X$ and $Y$ (Figure~\ref{fig:5Fcontra3}), since $L_j$ and $R_j$ do not overlap with $P$, and $L_j\oplus P$ and $R_j\oplus P$ do not include $Q$. If the number of consecutive cells to the left and right of $Y$ is $n$, the same argument applies sequentially to all cells to the left of $Y$ and to the right of $X$, establishing that there are no cells above or below them, so $P$ is a polyomino of height 1 and width $n$, which contradicts the assumption that $P$ is sufficiently large.
  \end{itemize}
\end{proof}

\begin{lemma}\label{lem:cover63}
  Hexomino~3 (Figure~\ref{fig:stain63}) is always-coverable.
\end{lemma}

\begin{figure}[ht]
  \centering
  \begin{minipage}{.40\textwidth}
    \centering
    \includegraphics[width=.9\linewidth]{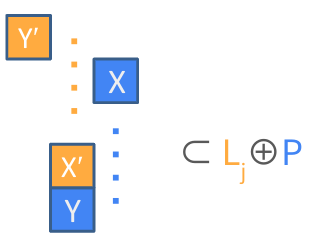}
    \captionof{figure}{}
    \label{fig:63setup1}
  \end{minipage}
  \begin{minipage}{.40\textwidth}
    \centering
    \includegraphics[width=.9\linewidth]{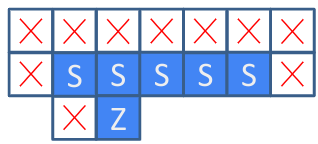}
    \captionof{figure}{}
    \label{fig:63setup2}
  \end{minipage}
\end{figure}

\begin{figure}[ht]
  \centering
  \begin{minipage}{.15\textwidth}
    \centering
    \includegraphics[width=.9\linewidth]{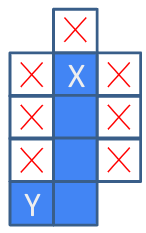}
    \captionof{figure}{}
    \label{fig:63contra1}
  \end{minipage}
  \begin{minipage}{.15\textwidth}
    \centering
    \includegraphics[width=.9\linewidth]{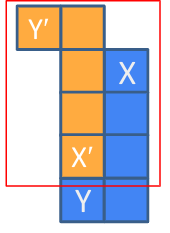}
    \captionof{figure}{}
    \label{fig:63contra2}
  \end{minipage}
  \begin{minipage}{.22\textwidth}
    \centering
    \includegraphics[width=.9\linewidth]{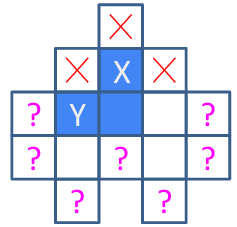}
    \captionof{figure}{}
    \label{fig:63contra3}
  \end{minipage}
  \begin{minipage}{.19\textwidth}
    \centering
    \includegraphics[width=.9\linewidth]{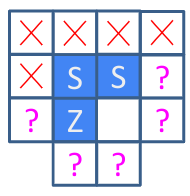}
    \captionof{figure}{}
    \label{fig:63contra4}
  \end{minipage}
  \begin{minipage}{.24\textwidth}
    \centering
    \includegraphics[width=.9\linewidth]{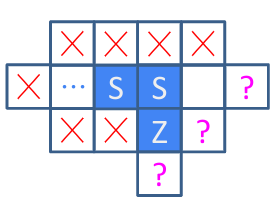}
    \captionof{figure}{}
    \label{fig:63contra5}
  \end{minipage}
\end{figure}

\begin{proof}
  Let $Q$ be a hexomino~3 and assume that $Q$ cannot be covered by a sufficiently large polyomino $P$. Let $L_0$ (resp. $R_0$) be the polyomino obtained by reflecting $P$ about its top side and shifting it to the left (resp. right) by one square, and let $L_i, R_i$ be polyominoes obtained by shifting $L_0, R_0$ down by $i$ squares, respectively. Take the smallest $j$ such that $L_{j+1}\cap P$ is non-empty and the smallest $k$ such that $R_{k+1}\cap P$ is non-empty. From symmetry, we may assume that $j\le k$. That is, none of $L_0, \cdots L_j, R_0, \cdots, R_j$ overlaps $P$. From the fact that $L_{j+1}$ overlaps $P$, we can take $X'\in L_j$ and $Y\in P$ such that $Y$ is one square below $X'$. Let $X$ be the cell in $P$ corresponding to $X'$ (Figure~\ref{fig:63setup1}). In the following, we divide cases according to the difference between the row index of $X$ and the row index of $Y$ in $P$.

  \begin{itemize}
    \item If the row index of $X$ is smaller than the row index of $Y$, $P$ does not contain the squares marked by $\times$ in Figure~\ref{fig:63contra1}, since $L_j$ does not contain the square under $X'$, and $L_0,\cdots, L_j, R_0, \cdots, R_j$ do not overlap $P$. Therefore, by connectedness, $P$ contains all squares below $X$ up to the right of $Y$. If the row index of $X$ is smaller than the row index of $Y - 1$, $L_j\oplus P$ includes $Q$ (Figure~\ref{fig:63contra2}), thus it must hold that the row index of $X$ equals the row index of $Y - 1$. In this case, there is a cell at one of the positions marked by \Qpink in Figure~\ref{fig:63contra3} that connects to $X$ without passing through any other \Qpink, but in any case, $L_j\oplus P$ or $R_j\oplus P$ includes $Q$, which is a contradiction.
    \item If the row index of $X$ is greater than the row index of $Y$, by reversing up and down and reversing the roles of $L_j$ and $P$, a contradiction can be derived in the same way.
    \item If the row index of $X$ equals the row index of $Y$, let $S$ be the set of cells that lie consecutively to the left and right of $Y$. Since $L_j$ and $R_j$ do not overlap $P$, $P$ does not contain any square that is the upper left or upper right of a cell in $S$. Since $P$ is not of height 1, it contains at least one square below $S$. Let $Z$ be the leftmost of them (Figure~\ref{fig:63setup2}). If $Z$ is adjacent to the leftmost cell in $S$, then there is a cell in one of the positions marked by \Qpink in Figure~\ref{fig:63contra4} and connects to $Z$ without passing through any other \Qpink, but in any case, $L_j\oplus P$ or $R_j\oplus P$ includes $Q$, which is a contradiction. Thus, $Z$ is adjacent to a non-leftmost cell of $S$. In this case, there is a cell in one of the positions marked by \Qpink in Figure~\ref{fig:63contra5} that connects to $Z$ without passing through any other \Qpink, but in any case, $L_j\oplus P$ or $R_j\oplus P$ includes $Q$, which is a contradiction.
  \end{itemize}
\end{proof}

\begin{lemma}\label{lem:cover6N}
  Hexomino~N (Figure~\ref{fig:stain6N}) is always-coverable.
\end{lemma}

\begin{figure}[ht]
  \centering
  \begin{minipage}{.30\textwidth}
    \centering
    \includegraphics[width=.9\linewidth]{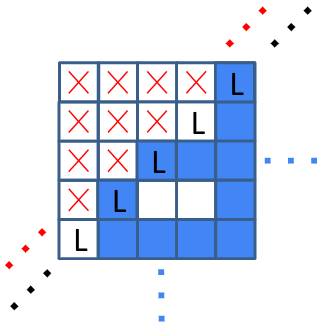}
    \captionof{figure}{}
    \label{fig:6Nsetup1}
  \end{minipage}
  \begin{minipage}{.30\textwidth}
    \centering
    \includegraphics[width=.9\linewidth]{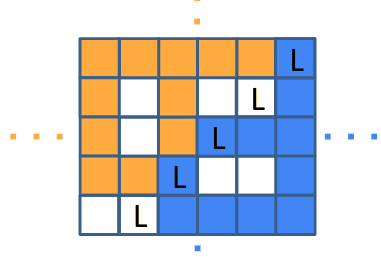}
    \captionof{figure}{}
    \label{fig:6Nsetup2}
  \end{minipage}
  \centering
  \begin{minipage}{.30\textwidth}
    \centering
    \includegraphics[width=.9\linewidth]{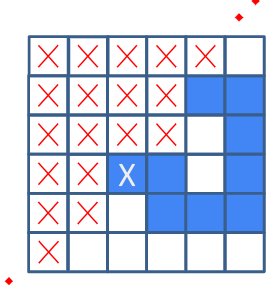}
    \captionof{figure}{}
    \label{fig:6Nsetup3}
  \end{minipage}
\end{figure}

\begin{figure}[ht]
  \centering
  \begin{minipage}{.23\textwidth}
    \centering
    \includegraphics[width=.9\linewidth]{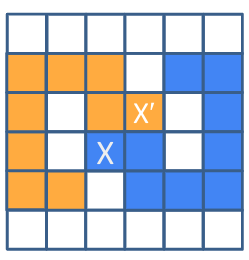}
    \captionof{figure}{}
    \label{fig:6Nsetup4}
  \end{minipage}
  \begin{minipage}{.23\textwidth}
    \centering
    \includegraphics[width=.9\linewidth]{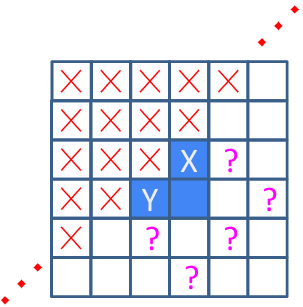}
    \captionof{figure}{}
    \label{fig:6Ncontra1}
  \end{minipage}
  \begin{minipage}{.23\textwidth}
    \centering
    \includegraphics[width=.9\linewidth]{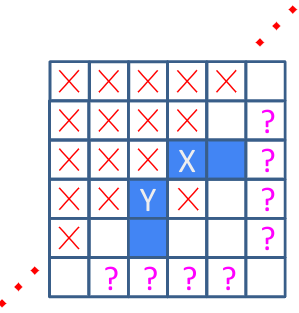}
    \captionof{figure}{}
    \label{fig:6Ncontra2}
  \end{minipage}
  \begin{minipage}{.23\textwidth}
    \centering
    \includegraphics[width=.9\linewidth]{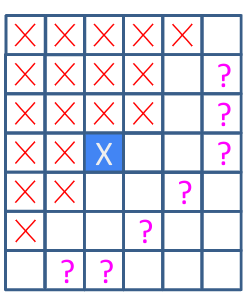}
    \captionof{figure}{}
    \label{fig:6Ncontra3}
  \end{minipage}
\end{figure}

\begin{proof}
  Let $Q$ be a hexomino~N and assume that $Q$ cannot be covered by a sufficiently large polyomino $P$. Let $L\subset\bbR^2$ be the uppermost 45 degree line of the shape of / that has a non-empty intersection with $P$ (Figure~\ref{fig:6Nsetup1}). All the following polyominoes do not overlap $P$.

  \begin{itemize}
    \item Polyomino $P_{\mathrm{L}}$, one which is obtained by reflecting $P$ about $L$ and shifting it by one square left (Figure~\ref{fig:6Nsetup2}).
    \item Polyomino $P_{\mathrm{U}}$, one which is obtained by reflecting $P$ about $L$ and shifting it by one square up.
    \item Polyomino $P_{\mathrm{L2}}$, one which is obtained by shifting $P_{\mathrm{L}}$ by one square left and down.
    \item Polyomino $P_{\mathrm{U2}}$, one which is obtained by shifting $P_{\mathrm{U}}$ by one square right and up.
  \end{itemize}

  We show the following property first.

  \begin{itemize}
    \item[$\star$] There are no two cells in $P\cap L$ that are diagonally consecutive (with Euclidean distance $\sqrt{2}$).
  \end{itemize}

  Assume that there are diagonally consecutive two cells, with $X$ at the upper right and $Y$ at the lower left. We can say that there is no cell below $X$ (Proof: Assume that there is a cell below $X$. There is a cell at one of the positions marked by \Qpink in Figure~\ref{fig:6Ncontra1} and connects to $X$ without passing through any other \Qpink, but in any case, one of $P_{\mathrm{L}}\oplus P, P_{\mathrm{U}}\oplus P, P_{\mathrm{L2}}\oplus P, P_{\mathrm{U2}}\oplus P$ includes $Q$). Therefore, there are cells to the right of $X$ and below $Y$. There is a cell at one of the positions marked by \Qpink in Figure~\ref{fig:6Ncontra2} and connects to $X$ without passing through any other \Qpink, but in any case, one of $P_{\mathrm{L}}\oplus P, P_{\mathrm{U}}\oplus P, P_{\mathrm{L2}}\oplus P, P_{\mathrm{U2}}\oplus P$ includes $Q$, which is a contradiction. Therefore, $\star$ has been proven.

  From $\star$, the following polyominoes also do not overlap with $P$.

  \begin{itemize}
    \item Polyomino $P_{\mathrm{LD}}$, one which is obtained by shifting $P_{\mathrm{L}}$ by one square down.
    \item Polyomino $P_{\mathrm{UR}}$, one which is obtained by shifting $P_{\mathrm{U}}$ by one square right.
  \end{itemize}

  Let $X$ be the bottom leftmost cell in $L\cap P$ (Figure~\ref{fig:6Nsetup3}). The following polyominoes also do not overlap with $P$.

  \begin{itemize}
    \item Polyomino $L_0$, one which is obtained by rotating $P$ 180 degrees around the lower left point of the square $X$.
    \item Polyomino $L_1$, one which is obtained by shifting $L_0$ by one square up.
    \item (From $\star$) Polyomino $R_0$, one which is obtained by rotating $P$ 180 degrees around the upper right point of the square $X$ (Figure~\ref{fig:6Nsetup4}).
    \item Polyomino $R_1$, one which is obtained by shifting $R_0$ by one square right.
  \end{itemize}

  There is a cell in one of the positions marked by \Qpink in Figure~\ref{fig:6Ncontra3} and connects to $X$ without passing through any other \Qpink, but in any case, the union of $P$ and one of the polyominoes mentioned above as a polyomino that does not overlap with $P$ includes $Q$, which is a contradiction.
\end{proof}

\begin{lemma}\label{lem:cover6S}
  Hexomino~S (Figure~\ref{fig:stain6S}) is always-coverable.
\end{lemma}

\begin{figure}[ht]
  \centering
  \begin{minipage}{.20\textwidth}
    \centering
    \includegraphics[width=.9\linewidth]{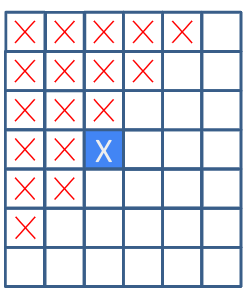}
    \captionof{figure}{}
    \label{fig:6Ssetup1}
  \end{minipage}
  \begin{minipage}{.20\textwidth}
    \centering
    \includegraphics[width=.9\linewidth]{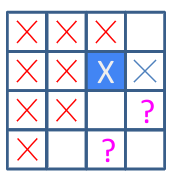}
    \captionof{figure}{}
    \label{fig:6Scontra1}
  \end{minipage}
  \begin{minipage}{.20\textwidth}
    \centering
    \includegraphics[width=.9\linewidth]{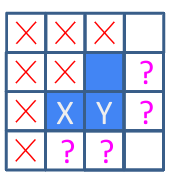}
    \captionof{figure}{}
    \label{fig:6Scontra2}
  \end{minipage}
  \begin{minipage}{.20\textwidth}
    \centering
    \includegraphics[width=.9\linewidth]{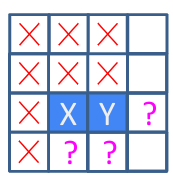}
    \captionof{figure}{}
    \label{fig:6Scontra3}
  \end{minipage}
\end{figure}

\begin{proof}
  Let $Q$ be a hexomino~S, and assume that $Q$ cannot be covered by a sufficiently large polyomino $P$. Let $L\subset\bbR^2$ be the uppermost 45 degree line of the shape of / that has a non-empty intersection with $P$, and let $X$ be the leftmost cell of $L\cap P$ (Figure~\ref{fig:6Ssetup1}). The following polyominoes do not overlap $P$.

  \begin{itemize}
    \item Polyomino $L_0$, one which is obtained by rotating $P$ 180 degrees around the lower left point of the square $X$.
    \item Polyomino $P_{\mathrm{L}}$, one which is obtained by reflecting $P$ about $L$ and shifting it by one square left.
    \item Polyomino $P_{\mathrm{U}}$, one which is obtained by reflecting $P$ about $L$ and shifting it by one square up.
  \end{itemize}

  There is a cell to the right of $X$ (Proof: otherwise, there would be a cell in one of the positions marked by \Qpink in Figure~\ref{fig:6Scontra1} and connects to $X$ without passing through any other \Qpink, but in any case $L_0\oplus P$ includes $Q$, which is a contradiction). Let $Y$ denote the cell. There is no cell above $Y$ (Proof: assuming there is a cell above $Y$, there is a cell in one of the positions marked by \Qpink in Figure~\ref{fig:6Scontra2}, but in any case $P_{\mathrm{L}}\oplus P$ or $P_{\mathrm{U}}\oplus P$ includes $Q$, which is a contradiction). Therefore, the following polyomino also does not overlap with $P$.

  \begin{itemize}
    \item Polyomino $R_0$, one which is obtained by rotating $P$ 180 degrees around the upper right point of the square $X$.
  \end{itemize}

  There is a cell in one of the positions marked by \Qpink in Figure~\ref{fig:6Scontra3}, but in any case $R_0\oplus P$ includes $Q$, which is a contradiction.
\end{proof}

\section{Polyominoes for Proof of Theorem~\ref{thm:uncoverable}}\label{sec:uncoverable}

The table below shows each polyomino in $\mathcal{I}$ on the left, and a polyomino that cannot cover it (counterexample) on the right: the two numbers in the first row indicate the height and width of the polyomino, and \verb|#| indicates the cells of the polyomino.

The counterexamples in this appendix are the smallest we found under our search pipeline and constraints; we do not prove that no smaller counterexample exists.

\begin{center}
  \begin{longtable}[!t]{ | p{3em} | p{39em}  | }
    \hline
\begin{verbatim}
1 5
#####
\end{verbatim} & \tiny
\begin{verbatim}
33 33
##...##.....####..####.....##...#
.##...##...##..####..##...##...##
..##...##.......#........##...##.
...##...##......#.......##...##..
#...##...##.....##.....##...##...
##...##...##.....#....##...##...#
.##...##...##....#...##...##...##
..##...##...##..##..##...##...##.
...##...##...##.#..##...##...##..
....##...##...#.#..#...##...##...
.#...##...##..###.##..##...##....
##....##...##..#..#..##...##...#.
#......##...####..####...##....##
#.......###.#..####.#...##......#
##........###.......#.###.......#
.###.####..#.......####........##
.#.###..####.......#..####..####.
##........#.......##.....####..#.
#.......###.......#..###.......##
#......##....####.####.##.......#
##....##...###..###.#...##......#
.#...##...##.##..#..##...##....##
....##...##...#.###..##...##...#.
...##...##...##.#.#...##...##....
..##...##...##..#.##...##...##...
.##...##...##..##..##...##...##..
##...##...##...#....##...##...##.
#...##...##....#.....##...##...##
...##...##.....##.....##...##...#
..##...##.......#......##...##...
.##...##........#.......##...##..
##...##...##..####..##...##...##.
#...##.....####..####.....##...##
\end{verbatim} \\
    \hline
\begin{verbatim}
2 3
###
#.#
\end{verbatim} & \tiny
\begin{verbatim}
71 71
.................................##.##.................................
..............................##..#.#..##..............................
...........................##..#..#.#..#..##...........................
...........................#...#..#.#..#...#...........................
...........................#...####.#####..#...........................
.......................###.#####..#..#..####.###.......................
.........................#..#.....#..#....#..#.........................
.....................##..#..#.....####....#..#..##.....................
...................#..#######.............#######..#...................
...................#..#.........................#..#...................
...................####.........................####...................
................##.#...............................#.##................
................#..#...............................#..#................
................#..#...............................#..#................
................####...............................####................
................#.....................................#................
...........######.....................................######...........
...........#..#.........................................#..#...........
..............#.........................................#..............
........#######.........................................#######........
..........#.................................................#..........
.......#..#.................................................#..#.......
.......####.................................................####.......
.....#..#.....................................................#..#.....
.....#..#.....................................................#..#.....
.....####.....................................................####.....
........#.....................................................#........
..####..#.....................................................#..####..
..#..####.....................................................####..#..
......#..........................................................#.....
.#....#..........................................................#...#.
.###..#..........................................................#####.
...####...........................................................#....
#..#..#...........................................................#...#
####..#.........................................................#######
......#.........................................................#......
#######.........................................................#..####
#..#............................................................####..#
...#............................................................#..#...
.#####..........................................................#..###.
.#...#..........................................................#....#.
.....#..........................................................#......
..#..####.....................................................####..#..
..####..#.....................................................#..####..
........#.....................................................#........
.....####.....................................................####.....
.....#..#.....................................................#..#.....
.....#..#.....................................................#..#.....
.......####.................................................####.......
.......#..#.................................................#..#.......
..........#.................................................#..........
........#######.........................................#######........
..............#.........................................#..............
...........#..#.........................................#..#...........
...........######.....................................######...........
................#.....................................#................
................####...............................####................
................#..#...............................#..#................
................#..#...............................#..#................
................##.#...............................#.##................
...................####.........................####...................
...................#..#.........................#..#...................
...................#..#######.............#######..#...................
.....................##..#..#.............#..#..##.....................
.........................#..#.............#..#.........................
.......................###.#####.......#####.###.......................
...........................#...#.......#...#...........................
...........................#...####.####...#...........................
...........................##..#..#.#..#..##...........................
..............................##..#.#..##..............................
.................................##.##.................................
\end{verbatim} \\
    \hline
\begin{verbatim}
3 3
###
#..
#..
\end{verbatim} &
    \scalebox{0.15}{\parbox{\textwidth}{\verbatiminput{\datadir solution/5/V.txt}}} \\
    \hline
\begin{verbatim}
3 3
.#.
###
.#.
\end{verbatim} & \tiny
\begin{verbatim}
9 9
...#.#...
..##.##..
#..#.#..#
#########
....#....
#########
#..#.#..#
..##.##..
...#.#...
\end{verbatim} \\
    \hline
    \verbatiminput{\datadir problem/5/Z.yes} & \tiny\verbatiminput{\datadir solution/5/Z.txt} \\
    \hline
\begin{verbatim}
5 2
#.
#.
##
.#
.#
\end{verbatim} &
    \scalebox{0.15}{\parbox{\textwidth}{\verbatiminput{\datadir solution/6/5.txt}}} \\
    \hline
\begin{verbatim}
4 3
.#.
.##
##.
.#.
\end{verbatim} & \tiny
\begin{verbatim}
17 17
.....#..#..#.....
.....#..#..#.....
....###.#.###....
......#####......
..#.....#.....#..
###.....#.....###
..##....#....##..
...#....#....#...
#################
...#....#....#...
..##....#....##..
###.....#.....###
..#.....#.....#..
......#####......
....###.#.###....
.....#..#..#.....
.....#..#..#.....
\end{verbatim} \\
    \hline
\begin{verbatim}
4 3
..#
###
.#.
.#.
\end{verbatim} & \tiny
\begin{verbatim}
21 21
.......##.#.##.......
........#.#.#........
......#.#.#.#.#......
.....###########.....
..........#..........
...#......#......#...
..##......#......##..
#..#......#......#..#
####......#......####
...#......#......#...
#####################
...#......#......#...
####......#......####
#..#......#......#..#
..##......#......##..
...#......#......#...
..........#..........
.....###########.....
......#.#.#.#.#......
........#.#.#........
.......##.#.##.......
\end{verbatim} \\
    \hline
\begin{verbatim}
3 4
##..
.###
..#.
\end{verbatim} & \tiny
\begin{verbatim}
21 21
.......#.....#.......
.......#.....#.......
.....###########.....
..........#..........
.....###########.....
..#.##.........##.#..
..#..#.........#..#..
###..#.........#..###
..#..#.........#..#..
..#..#.........#..#..
..####.........####..
..#............#..#..
..#............#..#..
###..#.........#..###
..#..#.........#..#..
..#.##.........##.#..
.....###########.....
..........#..........
.....###########.....
.......#.....#.......
.......#.....#.......
\end{verbatim} \\
    \hline
\begin{verbatim}
2 4
#..#
####
\end{verbatim} & \tiny
\begin{verbatim}
27 27
.........##.#.#.##.........
..........#.#.#.#..........
..........#.#.#.#..........
.......##.#.#.#.#.#........
........#.#######.#........
........#.#.....#.#........
........###.....###........
........#........#.....#...
...####.#........#######...
#.....###...........#.....#
#######.............#######
....#.................#....
#####.................#####
....#.................#....
#####.................#####
....#.................#....
#####...............#######
#.................###.....#
...#######........#.####...
...#.....#........#........
........###.....###........
........#.#.....#.#........
........#.#######.#........
........#.#.#.#.#.##.......
..........#.#.#.#..........
..........#.#.#.#..........
.........##.#.#.##.........
\end{verbatim} \\
    \hline
\begin{verbatim}
2 4
.##.
####
\end{verbatim} & \tiny
\begin{verbatim}
23 23
...........#...........
........##.#.##........
.........#.#.#.........
.....###.#####.###.....
.......###...###.......
...#...........#...#...
...#...........#...#...
...##..........#####...
.#..#.............#..#.
.####.............####.
...#...............#...
####...............####
...#...............#...
.####.............####.
.#..#.............#..#.
...#####.......#####...
...#...#.......#...#...
...#...#.......#...#...
.......###...###.......
.....###.#####.###.....
.........#.#.#.........
........##.#.##........
...........#...........
\end{verbatim} \\
    \hline
\begin{verbatim}
4 3
.##
##.
.#.
.#.
\end{verbatim} & \tiny
\begin{verbatim}
41 41
........#..#.....#.....#.....#..#........
.......#########.###.###.#########.......
...............#...#.#...#.........#.....
.....#############################.##....
...#...#.........................#.#.....
..######.........................###.#...
.....#.............................#.#...
.#.###.............................###.#.
##.#.................................#.##
.#.#.................................#.#.
.#.#.................................#.#.
##.#.................................#.##
.#.#.................................#.#.
.#.#.................................#.#.
.#.#.................................#.#.
.###.................................###.
...#.................................#...
##.#.................................#.##
.#.#.................................#.#.
.###.................................###.
...#.................................#...
.###.................................###.
.#.#.................................#.#.
##.#.................................#.##
...#.................................#...
.###.................................###.
.#.#...................................#.
.#.#...................................#.
.#.#...................................#.
##.#...................................##
.#.#...................................#.
.#.#...................................#.
##.#...................................##
.#.###.................................#.
...#.#...................................
...#.###.........................######..
.....#.#.........................#...#...
....##.#############################.....
.....#.........#...#.#...#...............
.......#########.###.###.#########.......
........#..#.....#.....#.....#..#........
\end{verbatim} \\
    \hline
\begin{verbatim}
4 2
##
##
#.
#.
\end{verbatim} & \tiny
\begin{verbatim}
61 61
..........................#.##.##.#..........................
.......................##.#..###..#.##.......................
........................###...#...###........................
....................#....#....#....#....#....................
.................##.#####################.##.................
...............#..###.........#.........###..#...............
...............#...#..........#..........#...#...............
...............##.##..........#..........##.##...............
...........###..###...........#...........###..###...........
.............####.............#.............####.............
..............................#..............................
........#.....................#.....................#........
........#.....................#.....................#........
........##....................#....................##........
.........#....................#....................#.........
.....###.#....................#....................#.###.....
.......###....................#....................###.......
....#...#.....................#.....................#...#....
....##.##.....................#.....................##.##....
.....###......................#......................###.....
...###........................#........................###...
....#.........................#.........................#....
....#.........................#.........................#....
.#..#.........................#.........................#..#.
.##.#.........................#.........................#.##.
..###.........................#.........................###..
###.#.........................#.........................#.###
....#.........................#.........................#....
#...#.........................#.........................#...#
##..#.........................#.........................#..##
.###########################################################.
##..#.........................#.........................#..##
#...#.........................#.........................#...#
....#.........................#.........................#....
###.#.........................#.........................#.###
..###.........................#.........................###..
.##.#.........................#.........................#.##.
.#..#.........................#.........................#..#.
....#.........................#.........................#....
....#.........................#.........................#....
...###........................#........................###...
.....###......................#......................###.....
....##.##.....................#.....................##.##....
....#...#.....................#.....................#...#....
.......###....................#....................###.......
.....###.#....................#....................#.###.....
.........#....................#....................#.........
........##....................#....................##........
........#.....................#.....................#........
........#.....................#.....................#........
..............................#..............................
.............####.............#.............####.............
...........###..###...........#...........###..###...........
...............##.##..........#..........##.##...............
...............#...#..........#..........#...#...............
...............#..###.........#.........###..#...............
.................##.#####################.##.................
....................#....#....#....#....#....................
........................###...#...###........................
.......................##.#..###..#.##.......................
..........................#.##.##.#..........................
\end{verbatim} \\
    \hline
\begin{verbatim}
3 3
###
##.
.#.
\end{verbatim} & \tiny
\begin{verbatim}
31 31
.............##.##.............
..........##.#...#.##..........
...........#########...........
............#.....#............
........#...#.....#...#........
.......######.....######.......
.........#...........#.........
.....#..##...........#...#.....
....##..#............###.##....
.....####..............###.....
.#...#...................#...#.
.##..#...................#..##.
..####...................####..
###.........................###
#.#.........................#.#
..#.........................#..
#.#.........................#.#
###.........................###
..#......................####..
.##......................#..##.
.#.......................#...#.
.....###..............####.....
....##.###............#..##....
.....#...#...........##..#.....
.........#...........#.........
.......######.....######.......
........#...#.....#...#........
............#.....#............
...........#########...........
..........##.#...#.##..........
.............##.##.............
\end{verbatim} \\
    \hline
\begin{verbatim}
4 3
###
.#.
.#.
.#.
\end{verbatim} & \tiny
\begin{verbatim}
23 23
.........##.##.........
..........###..........
......##...#...##......
.......#.#####.#.......
......##.#...#.##......
.......###...###.......
..#.#.##.......##.#.#..
..#####.........#####..
.....#...........#.....
#..###...........###..#
##.#...............####
.###................#..
##.#...............####
#..###...........###..#
.....#...........#.....
..####..........#####..
..#.#..........##.#.#..
.......###...###.......
......##.#...#.##......
.......#.##.##.#.......
......##..###..##......
..........#.#..........
.........##.##.........
\end{verbatim} \\
    \hline
\begin{verbatim}
4 3
#..
##.
.##
..#
\end{verbatim} & \tiny
\begin{verbatim}
17 17
......#....#.....
..#...#....#.....
..#...#....#####.
..#...#....#.....
..#...#....#.....
#######....#..#..
..#........######
..#...........#..
..#...........#..
..#...........#..
######........#..
..#..############
.....#....#...#..
.....#....#...#..
.#####....#...#..
.....#....#...#..
.....#....#......
\end{verbatim} \\
    \hline
    \verbatiminput{\datadir problem/6/Z.yes} & \tiny\verbatiminput{\datadir solution/6/Z.txt} \\
    \hline
\begin{verbatim}
4 4
###.
..##
...#
...#
\end{verbatim} & \tiny
\begin{verbatim}
17 17
....##.#.#.##....
...##..###..##...
....##..#..##....
.#...##.#.##...#.
###...#####...###
#.##......#..##.#
...##.....####...
##..#.......#..##
.####.......####.
##..#.......#..##
...####...####...
#.##..#...#..##.#
###...#####...###
.#...##.#.##...#.
....##..#..##....
...##..###..##...
....##.#.#.##....
\end{verbatim} \\
    \hline
  \end{longtable}
\end{center}

\end{document}